\documentclass[journal,12pt,onecolumn,draftclsnofoot]{IEEEtran}
\usepackage{ifpdf}

\usepackage{cite}

\ifCLASSINFOpdf
  \usepackage[pdftex]{graphicx}
\else
\fi
\usepackage{amsmath}
\usepackage{amsfonts}
\usepackage{amssymb}
\usepackage{amsthm}
\usepackage{color}
\usepackage[usenames,dvipsnames,svgnames,table]{xcolor}
\usepackage{algorithmic}
\usepackage{algorithm}
\usepackage{array}
\usepackage{url}
\usepackage{float}
\usepackage{graphicx}

\newtheorem{theorem}{Theorem}
\newtheorem{definition}{Definition}
\newtheorem{example}{Example}
\newtheorem{proposition}{Proposition}
\newtheorem{corollary}{Corollary}
\newtheorem{lemma}{Lemma}
\newtheorem{remark}{Remark}
\newtheorem{protocol}{Protocol}

\newtheorem{conjecture}{Conjecture}
\newcommand{\cV}{\mathcal{V}}
\newcommand{\cB}{\mathcal{B}}

\DeclareMathOperator*{\argmin}{arg\,min}

\begin{document}
%
\title{Communication-Efficient  Search for an Approximate Closest Lattice Point }

\author{Maiara~F.~Bollauf,
        Vinay~A.~Vaishampayan,
        and~Sueli~I.~R.~Costa.
\thanks{M. F. Bollauf is with the Institute of Mathematics, Statistics and Computer Science, University of Campinas as a Ph.D student (e-mail: maiarabollauf@ime.unicamp.br) and was a visiting scholar at  CUNY in 2016.}
\thanks{V. A. Vaishampayan is with Department of Engineering Science and Physics, City University of New York (CUNY) (e-mail: Vinay.Vaishampayan@csi.cuny.edu).}
\thanks{S. I. R. Costa is with the Institute of Mathematics, Statistics and Computer Science, University of Campinas (e-mail: sueli@ime.unicamp.br).}
\thanks{This paper was presented in part at the IEEE International Symposium on
Information Theory, Aachen, Germany, 2017 \cite{BVC:2017}.}}

\maketitle


\begin{abstract}
We consider the problem of finding the  closest lattice point to a vector in $n$-dimensional Euclidean space when each component of the vector is available at a distinct node in a network. Our objectives are (i)  minimize  the communication cost and (ii) obtain the error probability. The approximate closest lattice point considered here is the one obtained using the nearest-plane (Babai) algorithm.  Assuming a triangular special basis for the lattice, we develop  communication-efficient protocols for computing the approximate lattice point and determine the communication cost for lattices of dimension $n>1$. Based on available parameterizations of reduced bases, we determine   the error probability of the nearest plane algorithm for two dimensional lattices analytically, and present a computational error estimation algorithm in three dimensions. For dimensions 2 and 3, our results show that the error probability increases with the packing density of the lattice.  

\end{abstract}

{\small \textbf{\textit{Index terms}---Lattices, lattice quantization, distributed function computation, communication complexity.}}

\IEEEpeerreviewmaketitle

%

\section{Introduction}


	
	
	\begin{figure}[h]
\centering
\begin{minipage}{.35\linewidth}
  \includegraphics[width=\linewidth]{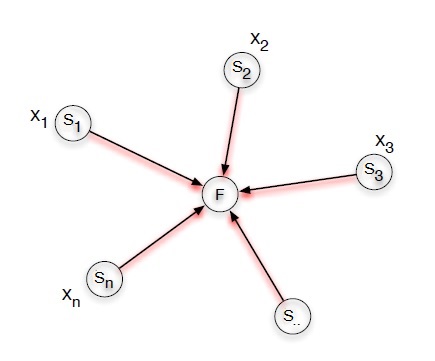}
  \caption{Centralized model}
  \label{fig-p1}
\end{minipage}
\hspace{.05\linewidth}
\begin{minipage}{.35\linewidth}
  \includegraphics[width=\linewidth]{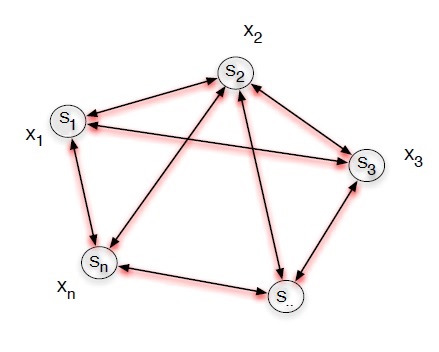}
  \caption{Interactive model}
  \label{fig-p2}
\end{minipage}
\end{figure}

We consider a network consisting of $N$ sensor-processor nodes (hereafter referred to as nodes) and possibly a central computing node (fusion center) $F$ interconnected by links  with limited bandwidth. 
Node $i$ observes real-valued random variable $X_i$. In the  \textit{centralized model} (Fig. \ref{fig-p1}), the objective is to compute a given function $f(X_1,X_2,\ldots,X_n)$ at the fusion center based on information communicated from each of the $N$ sensor nodes. In the \textit{interactive model} (Fig. \ref{fig-p2}), the objective is to compute the function $f(X_1,X_2,\ldots,X_n)$ at each sensor node (the fusion center is absent). In general, since the random variables are real valued, these calculations would require that the system communicate an infinite number of bits in order to compute $f$ exactly. Since the network has finite bandwidth links, the information must be quantized in a suitable manner, but quantization affects the accuracy of the function that we are trying to compute. Thus, the main goal is to manage the tradeoff between communication cost and function computation accuracy. Here, $f$ computes the closest lattice point  to a real vector $x=(x_1,x_2,\ldots,x_n)$ in a given lattice $\Lambda$. 

	The process of finding the closest lattice point is widely used for decoding lattice codes and for quantization.  Lattice coding  offers significant coding gains for noisy channel communication~\cite{conwaysloane}  and for quantization~\cite{berger}. In a network, it may be necessary for a vector of measurements to be available at locations other than and possibly including nodes where the measurements are made. In order to reduce network bandwidth usage, it is logical to consider a vector quantized (VQ) representation of these measurements, subject to a fidelity criterion, for once a VQ representation is obtained, it can be forwarded in a bandwidth efficient manner to other parts of the network. However, there is a communication cost to obtaining the vector quantized representation. This paper is our attempt to understand the costs and tradeoffs involved. Bounds for the error probability in dimensions 2 and 3, when we considered the nearest plane (Babai) algorithm for lattices, are derived as well as the rate computation which underlies the decoding process in a distributed system. The communication cost/error tradeoff  of refining the nearest-plane estimate in an interactive setting is addressed in a companion paper~\cite{VB:2017}.

	Example application settings include MIMO systems~\cite{RCP:2009}, and network management in wide area networks~\cite{KCR:2006}, to name a few. For prior work in the computer science community, see~\cite{Yao:1979},~\cite{KN:1997}. Information theory~[\ref{Shannon:1948}] has resulted in tight bounds,~\cite{Orlitsky:2001},~\cite{MaIshwar}.  

	We observe here that algorithms for the closest lattice point problem have been studied in great detail, see~\cite{agrelletal} and the references therein, for a comprehensive survey and algorithms. However, in all these algorithms it is assumed that the vector components are available at the same location. In our work, the vector components are available at physically separated nodes and we are interested in the communication cost of exchanging this information in order to determine the closest lattice point. None of the previously proposed fast algorithms consider this communication cost.

	The closest lattice point problem has also been proposed as a basis for lattice cryptography (\cite{Ajtai:1996},\cite{MicGold:2012},\cite{galbraith},~\cite{mathcrypto},\cite{peikert}), a topic of great interest  in recent years, examples being the GGH and LWE cryptosystems. The idea is to require solution of the closest lattice point problem, which is known to be NP-hard~\cite{emdeboas}, assuring security.  
The Babai algorithm is used in some of the proposed cryptosystems, thus computation of its error probability is of interest in this context.



The remainder of our paper is organized into five sections:  Sec.~\ref{sec1} presents some basic definitions and facts about lattices.  Sec.~\ref{sec3} establishes a framework for measuring the cost and error rate and presents an expression for the error probability of the distributed closest lattice point problem in an arbitrary two dimensional case, Sec.~\ref{sec3d} develops a computational error analysis procedure for   lattices in dimension 3. Sec.~\ref{sec:babairate} presents efficient protocols for computing the approximate closest lattice point along with rate estimates for both models for $n>1$. 
Directions for future work and conclusions are in Sec.~\ref{secCC}. 

\section{Lattice Basics, Voronoi and Babai Partitions} 
\label{sec1} 
Notation, properties, partitions and special bases for lattices are described in this section. 

A lattice $\Lambda \subset \mathbb{R}^n$ is a set of integer linear combinations of independent vectors $v_1,v_2,\ldots,v_M \in \mathbb{R}^n,$ which can be written as $\Lambda=\{Vu,~u\in \mathbb{Z}^M\},$ where $V$ is a matrix whose columns are the vectors $v_{1}, \dots, v_{M}$ and vectors are considered here in the column format. 
$V$ is a \textit{generator matrix} of $\Lambda$  and $A=V^{T}V$ is the associated \textit{Gram matrix}. In this paper, we only consider full rank lattices ($n=M$).
	
A set $\mathcal{F}$ is called a \textit{fundamental region} of a lattice $\Lambda$ if all its translations by elements of $\Lambda$ define a partition of $\mathbb{R}^{n}$. Examples of fundamental regions are the fundamental parallelepiped supported by a set of basis vectors and the \textit{Voronoi region} or \textit{Voronoi cell} $\cV(\lambda)$ of a lattice point defined as $\cV(\lambda)=\{x \in \mathbb{R}^{n}: ||x-\lambda|| \leq ||x-\tilde{\lambda}||, \ \text{for all} \ \tilde{\lambda} \in \Lambda\},$ where $||.||$ denotes the Euclidean norm. Note that  $\cV(\lambda)$ is congruent to $\cV(0)$. 	

The \textit{volume} of a lattice $\Lambda$ is the volume of any of its fundamental regions. It is given by $vol(\Lambda)=|\det(V)|,$ where $V$ is a generator matrix of $\Lambda.$
	
A vector $v$ is called \textit{Voronoi vector} if the hyperplane $\{x \in \mathbb{R}^{n}: x \cdot v=\frac{1}{2} v \cdot v\}$ has a non-empty intersection with $\mathcal{V}(0).$ A Voronoi vector is said to be \textit{relevant} if this intersection is an $(n-1)-$dimensional face of $\cV(0).$  



The packing radius $\rho$ of a lattice $\Lambda$ is half of the minimum distance between lattice points and the packing density $\Delta(\Lambda)$ is the fraction of  space that is covered by balls $\mathcal{S}(\lambda, \rho)$ of radius $\rho$ in $\mathbb{R}^{n}$ centered at lattice points, i.e., $\Delta(\Lambda)  =  \dfrac{vol \ S(0,\rho)}{vol (\Lambda)}.$

	The \textit{closest vector problem} (CVP) in a lattice can be described as an integer least squares problem with the objective of determining $u^*,$ such that $u^\ast =  \argmin_{u \in \mathbb{Z}^{n}} \mid\mid x-Vu \mid\mid^{2},$ where the norm considered is the standard Euclidean norm. The closest lattice point to $x$ is then given by $x_{nl}=Vu^\ast$. Observe that the mapping $g_{nl}~:~\mathbb{R}^n \rightarrow \Lambda, \ \  x \mapsto x_{nl}$ partitions $\mathbb{R}^n$ into  Voronoi cells. 

The nearest plane (np) algorithm~\cite{babai}, an approach for approximating the closest lattice point,  computes $x_{np}$, an approximation to $x_{nl}$, given by $x_{np}=b_1 v_1 + b_2 v_2+\ldots+b_n v_n$, where $b_i \in \mathbb{Z}$ is obtained as follows.

	Let ${\mathcal S}_i$ denote the subspace spanned by the vectors $\{v_1,v_2,\ldots,v_i\}$, $i=1,2,\ldots,n$. Let ${\mathcal P}_i(z)$ be the orthogonal projection of $z$ onto ${\mathcal S}_i$ and let $v_{i,i-1}={\mathcal P}_{i-1}(v_i)$ be the closest vector to $v_i$ in ${\mathcal S}_{i-1}$. Consider the decomposition $v_i=v_{i,i-1}+v_{i,i-1}^{\perp}$ and let $z_i^\perp=z_{i}-{\mathcal P}_{i}(z_{i})$. Start with $z_n=x$ and $i=n$ and compute $b_i=\left[ \langle z_{i},v_{i,i-1}^\perp \rangle/\|v_{i,i-1}^\perp\|^2 \right]$, $z_{i-1}={\mathcal P}_{i-1}(z_i)-b_i v_{i,i-1}$, for $i=n,n-1,\ldots,1$ (here $[x]$ denotes the nearest integer to $x$).
	
 	We denote the vector $b=(b_{1}, b_{2}, \dots, b_{n})$ as \textit{Babai point}, which is an approximate solution to the closest lattice point problem. The mapping $g_{nl}~:~\mathbb{R}^n \rightarrow \Lambda, \ \ x \mapsto x_{np}$ partitions $\mathbb{R}^n$ into  hyper-rectangular cells with volume $|\det V|$ and we refer to this partition as a \emph{Babai partition}.

\begin{example} Fig.~\ref{fig-np} represents the Babai partition (black lines) and the Voronoi partition (pink lines) for the hexagonal lattice $A_2$ generated by $\{(1,0),(1/2, \sqrt{3}/2)\}$ and illustrates geometrically the manner in which the   np algorithm approximates the closest point problem. 

\begin{figure}[h!]
\begin{center}
		\includegraphics[height=3.5cm]{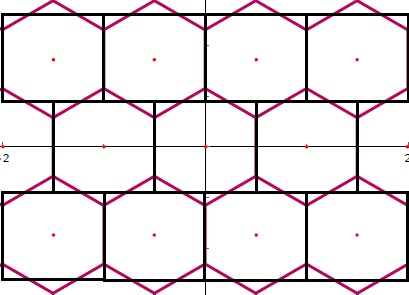}  
\caption{{Cells of the Babai partition and of the Voronoi partition for the hexagonal lattice $A_2$}}
 \label{fig-np}
\end{center}
\end{figure}
\end{example}
	Note that a Babai partition is basis dependent. In case the generator matrix $V$ is upper triangular with $(i,j)$ entry $v_{ij}$, each rectangular cell is axis-aligned and has sides of length  $|v_{11}|,|v_{22}|,\ldots,|v_{nn}|$. In this specific case, the vectors $v_{i,i-1}^{\perp}$ mentioned above are of type $(0, \dots, v_{ii},$ $0, \dots,0)$. 
	
	We remark that given a lattice $\Lambda$ with an arbitrary generator matrix $V$ we can apply the QR decomposition, $V=QR$ and the matrix $R$ will generate a rotation of the original lattice. We will introduce next two types of special bases we will work closely in this paper: Minkowski-reduced basis and obtuse superbase.
	
%

	 A basis $\{v_{1},v_{2},...,v_{n}\}$ of a lattice $\Lambda$
in $\mathbb{R}^{n}$ is said to be \textit{Minkowski-reduced}
if $v_{j},$ with $j=1,\dots,n,$ is such that $\left\Vert v_{j}\right\Vert \leq\left\Vert v\right\Vert $, for any $v$ for which $\{v_{1},...,v_{j-1},v\}$ can be extended
to a basis of $\Lambda$.


For two dimensional lattices, a Minkowski-reduced
basis is also called Lagrange-Gauss reduced basis and there
is a simple characterization~\cite{conwaysloane} for it: 
a lattice basis $\left\{ \ensuremath{v_{1},v_{2}}\right\} $
is a Minkowski-reduced basis if only if $\left\Vert v_{1}\right\Vert \leq\left\Vert v_{2}\right\Vert $
and $2\langle v_{1},v_{2}\rangle  \leq \left\Vert v_{1}\right\Vert ^{2}.$ 
It follows that the angle $\theta$ between
the subsequent minimum norm vectors $v_{1}$ and $v_{2}$ must satisfy $\text{ }\frac{\pi}{3} \leq\theta\leq\frac{2\pi}{3}.$


It is also possible to characterize a Minkowski-reduced basis for lattices in dimensions  three or smaller~\cite{conwaysloane} : 
\begin{proposition} \label{propmink} Consider a Gram matrix $A$ of a lattice $\Lambda$ and the conditions below:
\begin{eqnarray}
0 < a_{11}  & \leq & a_{22}  \leq a_{33}  \label{mink1} \\ 
2|a_{st}| & \leq & a_{ss} \ \ (s < t) \label{mink2} \\
2|\pm a_{rs} \pm a_{rt} \pm a_{st}| & \leq & a_{rr} + a_{ss} \ \ (r < s < t) \label{mink3}.
\end{eqnarray}
Then, inequalities (\ref{mink1}), (\ref{mink1})--(\ref{mink2}) and (\ref{mink1})--(\ref{mink3}) define a Minkowski-reduced basis for dimensions 1,2 and 3, respectively.
\end{proposition}
All lattices have a Minkowski-reduced basis, which  roughly speaking, consists of short vectors that are as  perpendicular as possible~\cite{conwaysloane}. 
	
We describe next the concept of an obtuse superbase that will be used in the following sections.
%


\begin{definition}
Let $\{v_1,v_{2}, \dots, v_{n}\}$ be a basis for a lattice $\Lambda$. A {\bf superbase}   $\{v_{0}, v_{1}, \dots, v_{n}\}$ with $v_0= -\sum_{i=1}^{n} v_{i},$  is said to be {\bf obtuse} if $p_{ij}=v_{i} \cdot v_{j} \leq 0,$ for $i,j=0,\dots,n, \ \ i \neq j$. A lattice $\Lambda$ is said to be of {\bf Voronoi's first kind} if it has an \textit{obtuse superbase}
\end{definition}
The above parameters $p_{ij}$ are called \textit{Selling parameters} and if $p_{ij}<0$ we say that the superbase is strictly obtuse.


\begin{example} \label{exos} Consider the standard basis $\{v_{1},v_{2},v_{3}\}$ for the body-centered cubic (BCC) lattice where $v_{1}=(1,1,-1), v_2=(1,-1,1), v_3=(-1,1,1).$ We set $v_0=-v_1-v_2-v_3=(-1,-1,-1)$ and it is not hard to see that $v_0,v_1,v_2,v_3$ is a strictly obtuse superbase for BCC lattice. Indeed $v_0+ v_1 + v_2 + v_3=0$ and $p_{ij}= -1 < 0 \ \forall i,j=0,1,2,3, i \neq j.$ Thus, BCC is of Voronoi's first kind.
\end{example}

The existence of an obtuse superbase allows a characterization of the relevant Voronoi vectors for a lattice.

\begin{theorem} \cite[Th.3, Sec. 2]{ConwaySloane:1992} Let $\Lambda$ be a lattice of Voronoi's first kind with obtuse superbase $v_0,v_1, \dots, v_n$. 
Vectors of the form $\sum_{i \in S} v_{i},$ where $S$ is a strict non-empty subset of $\{0,1,\dots, n\}$ are Voronoi vectors of $\Lambda$.
\end{theorem}

	It was demonstrated \cite{ConwaySloane:1992} that all lattices with dimension less or equal than three are Voronoi's first kind. In three dimensions, considering an obtuse superbase, since $v_0=-v_1-v_2-v_3,$ all Voronoi vectors described in the above theorem can be  written as one of the following seven vectors or their opposites \cite{ConwaySloane:1992}: 
\begin{equation}
v_{1},v_{2},v_{3},v_{12}=v_{1}+v_{2}, v_{13}=v_{1}+v_{3}, v_{23}=v_{2}+v_{3}, v_{123}=v_{1}+v_{2}+v_{3}
\end{equation}

	Given an obtuse superbase, we also characterize the norms $N(v_{1}),N(v_{2}),N(v_{3}),N(v_{12}),N(v_{13}),$ $N(v_{23}), N(v_{123}),$ where $N(x)=x \cdot x,$ as \textit{vonorms} and $p_{ij}=-v_{i} \cdot v_{j} (0 \leq i < j \leq 3)$ as \textit{conorms}, for the superbase $v_0,v_1,v_2,v_3.$ Precise definitions of conorms and vonorms for the general $n-$dimensional case can be found in \cite{ConwaySloane:1992}.

%
	The nonzero cosets of $\Lambda/2\Lambda$ naturally form a discrete projective plane of order $2.$ The vonorms 
are marked as the nodes of the projective plane and the corresponding conorms $0$ and $p_{ij}$ at the nodes of the dual plane in the following Figure \ref{gpp}.
	
\begin{figure}[H]
\begin{center}
		\includegraphics[height=4.5cm]{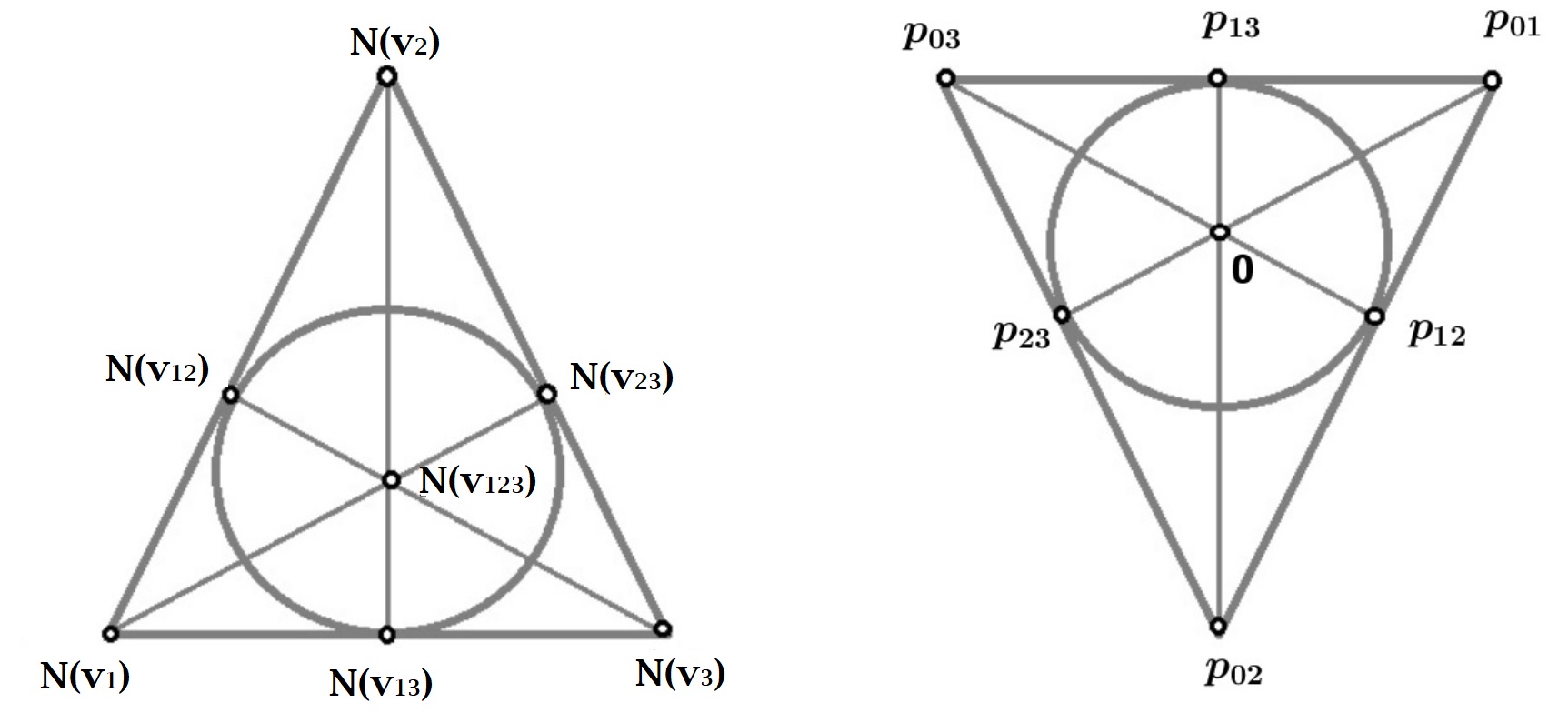}  
\caption{{Projective and dual planes labelled with vonorms and conorms respectively (based on \cite{ConwaySloane:1992}, p. 61)}}
 \label{gpp}
\end{center}
\end{figure}	

\begin{remark}\cite{ConwaySloane:1992} Two projective planes labeled with conorms represent the same lattice precisely when there is a conorm preserving collineation between them.
\end{remark}

%
%
%
	
	The existence of obtuse superbases for three dimensional lattices let us characterize the five parallelohedra that can be their Voronoi regions: truncated octahedron, hexa-rhombic dodecahedron, rhombic dodecahedron, hexagonal prism and cuboid.
	
	Let $\Lambda$ be an arbitrary $3-$dimensional lattice, with obtuse superbase $v_0,v_1,v_2,v_3$ and conorms $p_{i,j}.$ A vector $t \in \mathbb{R}^{3}$ can be specified by its inner products
\begin{equation} \label{eqr3}
(t \cdot v_1,t \cdot v_2,t \cdot v_3)=(y_1,y_2,y_3)=y,
\end{equation} 
since the determinant of the Gram matrix $A$ is non vanishing. 

	The most generic Voronoi region in three dimensions is the truncated octahedron, with 14 faces and 24 vertices. 
It is know \cite{ConwaySloane:1992} that the vertices of this Voronoi cell are all the 24 points $p_{ijkl}$ where $\{i,j,k,l\}$ is any permutation of $\{0,1,2,3\}$:
\begin{eqnarray} \label{eqcoor}
y_{i}=\small{\frac{1}{2}}(p_{ij}+p_{ik}+p_{il}), \ \ \ y_{j}=\small{\frac{1}{2}}(-p_{ji}+p_{jk}+p_{jl}), \nonumber \\
y_{k}=\small{\frac{1}{2}}(-p_{ki}-p_{kj}+p_{kl}), \ \ \ y_{l}=\small{\frac{1}{2}}(-p_{li}-p_{lj}-p_{lk}). 
\end{eqnarray}
 
	Using Equations (\ref{eqr3}) and (\ref{eqcoor}) one can define all the points that generates a generic Voronoi region. 
	
	It is possible to guarantee that two lattices for which the correspondent conorms are zero have combinatorially equivalent Voronoi regions (since one can be continuously deformed into the other without any edges being lost). So, when we construct the dual projective planes to represent the conorms, there are five choices for zeros: one, two, three collinear zeros, three non-collinear zeros or four zeros. Each of these configuration produces a different Voronoi cell according to Figure \ref{5types}.
	
\begin{figure}[H]
\begin{center}
		\includegraphics[height=5.0cm]{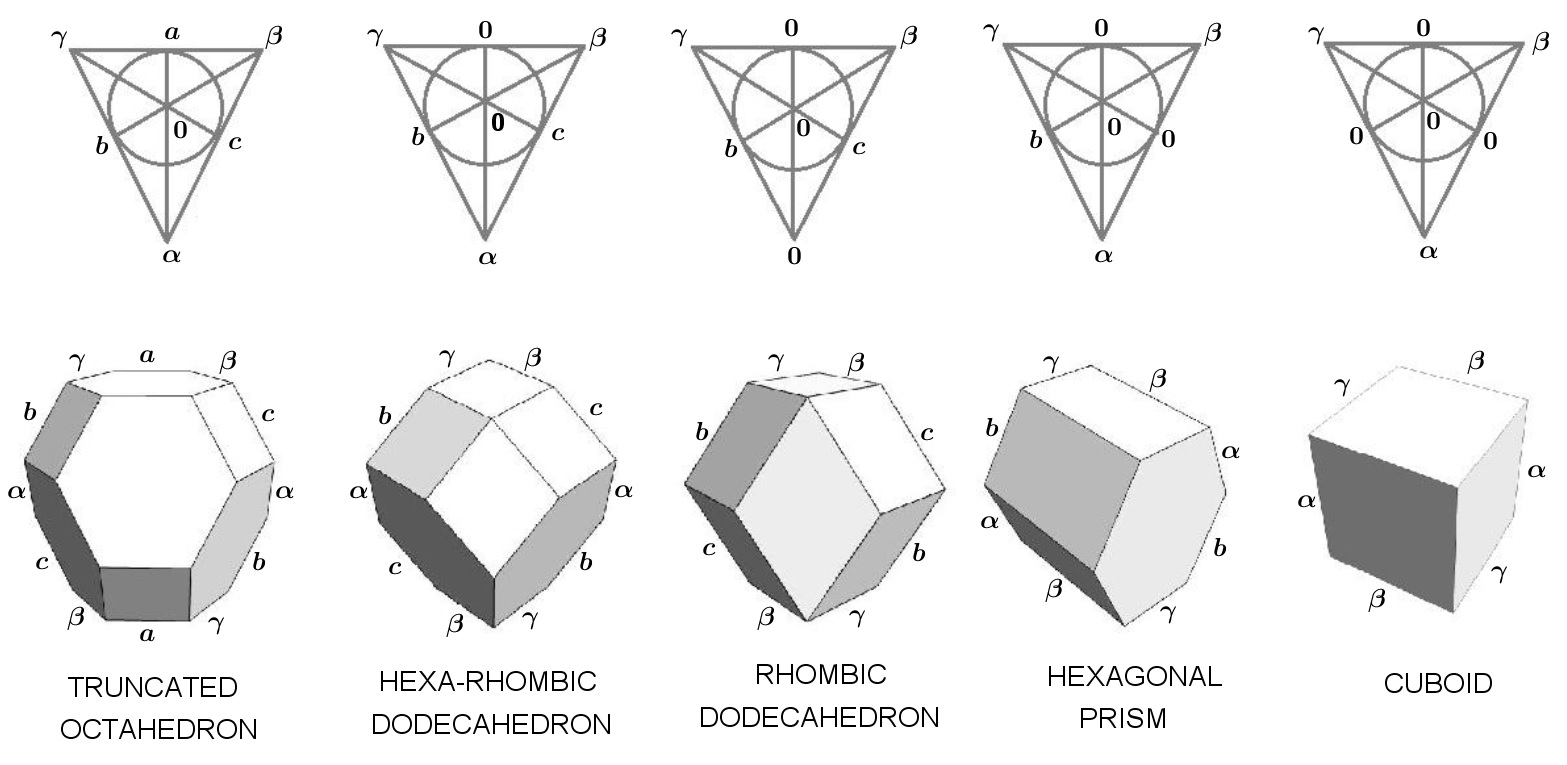}  
\caption{{Dual plane labeled with conorms and its correspondent Voronoi cells (based on \cite{ConwaySloane:1992}, p. 65). (\emph{Note: Authors have requested permission to reproduce from the publisher and will present only if permission is granted.})}}
 \label{5types}
\end{center}
\end{figure}


\begin{theorem} \label{minkos}
In dimensions $n=1,2,3$, if a lattice $\Lambda \subset \mathbb{R}^n$ has a  Minkowski-reduced basis with vectors $\{v_1,\ldots,v_n\}$, with $v_i . v_j \leq 0$, $i\neq j$, then the superbase $\{v_0,v_1,\ldots,v_n\}$, with $v_0=-\sum_{i=1}^n v_i$ is an obtuse superbase for $\Lambda$. Conversely, if $\Lambda$ has an obtuse superbase, then a Minkowski-reduced basis can be obtained from it.
\end{theorem}
\begin{proof} The case \fbox{\begin{minipage}{1.5em}
n=1 
\end{minipage}} 
is trivial, hence we will start with 
%
\fbox{\begin{minipage}{1.5em}
n=2: 
\end{minipage}} $(\Rightarrow)$ Suppose that $\{v_{1},v_{2}\}$ is a Minkowski-reduced basis, then, according to Proposition \ref{propmink}, $0< v_{1}\cdot v_1 \leq v_2 \cdot v_2$ and $2|v_1 \cdot v_2| \leq v_{1} \cdot v_1.$ Moreover, by hypothesis, $v_1 \cdot v_2 \leq 0.$ Define $v_0=-v_1-v_2$ and to guarantee that $\{v_0,v_1,v_2\}$ is an obtuse superbase, we need to check that $p_{01} \leq 0$ and $p_{02} \leq 0.$ Indeed,
\begin{eqnarray}
p_{01}= v_0 \cdot v_1 = (-v_1 -v_2) \cdot v_1 = -v_1 \cdot v_1 \underbrace{-v_1 \cdot v_2}_{|v_1 \cdot v_2|} \leq -2|v_1 \cdot v_2| + |v_1 \cdot v_2| \leq 0.
\end{eqnarray}
	Similarly we have that $p_{02} \leq 0.$

$(\Leftarrow)$ If $\{v_0,v_1,v_2\}$ is an obtuse superbase, any permutation of it is also an obtuse superbase. So, we may consider one such that $|v_1| \leq |v_2| \leq |v_0|.$ Then we have that $0 < v_1 \cdot v_1 \leq v_{2} \cdot v_{2} \leq (v_1+v_2) \cdot (v_1+v_2)$ and $v_{1} \neq 0.$



	From the last inequality, we have that
\begin{eqnarray}
-2v_{1} \cdot v_{2} & \leq  & v_{1} \cdot v_{1} \Rightarrow 2 |v_{1} \cdot v_{2}| \leq  v_{1} \cdot v_{1}.  
\end{eqnarray}

For
\fbox{\begin{minipage}{1.5em}
n=3: 
\end{minipage}} $(\Rightarrow)$ Consider a Minkowski-reduced basis $\{v_1,v_{2},v_{3}\}$ such that $v_{1}\cdot v_{2} \leq 0, v_{1} \cdot v_{3} \leq 0$ and $v_{2} \cdot v_{3} \leq 0.$ To check if $\{v_{0},v_1,v_2,v_{3}\}$ is an obtuse superbase, we need to verify that $p_{01} \leq 0, p_{02} \leq 0$ and $p_{03} \leq 0.$ 
	
	Observe that
\begin{equation}
p_{01}= v_0 \cdot v_{1} = -v_1 \cdot v_1 \underbrace{-v_1 \cdot v_2}_{|v_{1} \cdot v_{2}|} \underbrace{-v_1 \cdot v_3}_{|v_1 \cdot v_3|} \leq  -v_1 \cdot v_1 +  \frac{v_1 \cdot v_1}{2} + \frac{v_1 \cdot v_1}{2} \leq 0.
\end{equation}	

%

	With analogous arguments, we show that $p_{02} \leq 0$ and $p_{03} \leq 0.$

	
$(\Leftarrow)$ To prove the converse, up to a permutation, we may consider an obtuse superbase such that $|v_{1}| \leq |v_2| \leq |v_3| \leq |v_{0}|.$ This basis will be Minkowski-reduced if we prove conditions ($\ref{mink2}$) and ($\ref{mink3}$) from Proposition \ref{propmink}, i.e.,

\begin{equation} \label{ine1}
2|v_1 \cdot v_2| \leq v_1 \cdot v_1; \ \  2|v_1 \cdot v_3| \leq v_1 \cdot v_1; \ \  2|v_2 \cdot v_3| \leq v_2 \cdot v_2 
\end{equation}
and
\begin{equation} \label{ine2}
2|\pm v_1 \cdot v_2 \pm v_{1} \cdot v_{3} \pm v_{2} \cdot v_{3}| \leq v_1 \cdot v_1+ v_2 \cdot v_2.
\end{equation}

	The inequalities in Equation (\ref{ine1}) are shown similarly to the two dimensional case starting from $v_{2} \cdot v_{2} \leq (v_1+v_2) \cdot (v_1+v_2),$ $v_{3} \cdot v_{3} \leq (v_1+v_3) \cdot (v_1+v_3)$ and $v_{3} \cdot v_{3} \leq (v_2+v_3) \cdot (v_2+v_3).$ Starting from $v_{3} \cdot v_{3} \leq (v_1 + v_2 +v_3) \cdot (v_1 + v_2 + v_3),$ it follows the inequality in Equation (\ref{ine2}) concluding the proof.  

%
\end{proof}

\section{Error Probability Analysis for Arbitrary Two Dimensional Lattices}
\label{sec3}

	We  assume that node $i$ observes an independent identically distributed (iid) random process $\{X_i(t), t \in \mathbb{Z}\}$, where $t$ is the time index and that random processes observed at distinct nodes are mutually independent. The time index $t$ is suppressed in the sequel. The random vector $X=(X_1,X_2)$ is obtained by projecting a random process on the basis vectors of an underlying coordinate frame, which is assumed to be fixed.

	Consider that the lattice $\Lambda$ is generated by the scaled generator matrix $\alpha V$, where $V$ is the generator matrix of the unscaled lattice. 
Let ${\mathcal V}(\lambda)$ and ${\mathcal B}(\lambda)$ denote the Voronoi and Babai cells, respectively, associated with lattice vector $\lambda \in \Lambda$. The error probability $P_e(\alpha)$, is the probability of the event $\{\lambda_{nl}(X)\neq \lambda_{np}(X)\}$ and $P_e:=\lim_{\alpha \rightarrow 0} P_e(\alpha) = {area(\cB(0)\bigcap \cV(0)^c)}/{area(\cB(0))}$. 

	As will be discussed further, the Babai partition is dependent on, and the Voronoi partition is invariant to, the choice of lattice basis. Thus the error probability depends on the choice of the lattice basis. We will assume here that a Minkowski-reduced lattice basis, which is also obtuse (Theorem \ref{minkos}) can be chosen by the designer of the lattice code and it can be transformed into an equivalent basis $\{(1,0),(a,b)\}.$ This can be accomplished by applying QR decomposition to the lattice generator matrix (which has the original chosen basis vectors on its columns) in addition to convenient scalar factor. The reason for working with a Minkowski-reduced basis is partly justified by Example~\ref{ex1} below and the fact that the Voronoi region is easily determined since the relevant vectors are known; see Lemma~\ref{lemmathird} below.




An example to demonstrate the dependence of the error probability on the lattice basis is now presented.
\begin{example} \label{ex1} Consider a lattice
$\Lambda\subset\mathbb{R}^{2}$ with basis $\{(5,0),(3,1)\}.$ The error probability in this case is $P_{e}=0.6$ (Fig. \ref{triangbasis}), whereas if we
start from the basis $\{(1,2),(-2,1)\},$ we achieve after the QR decomposition $\left\{ (\sqrt{5},0),(0,\sqrt{5})\right\}$ and $P_{e}=0,$ since the Babai region associated
with an orthogonal basis and the Voronoi region for rectangular lattices coincides.

\begin{figure}[h!]
\begin{center}
		\includegraphics[height=3.0cm]{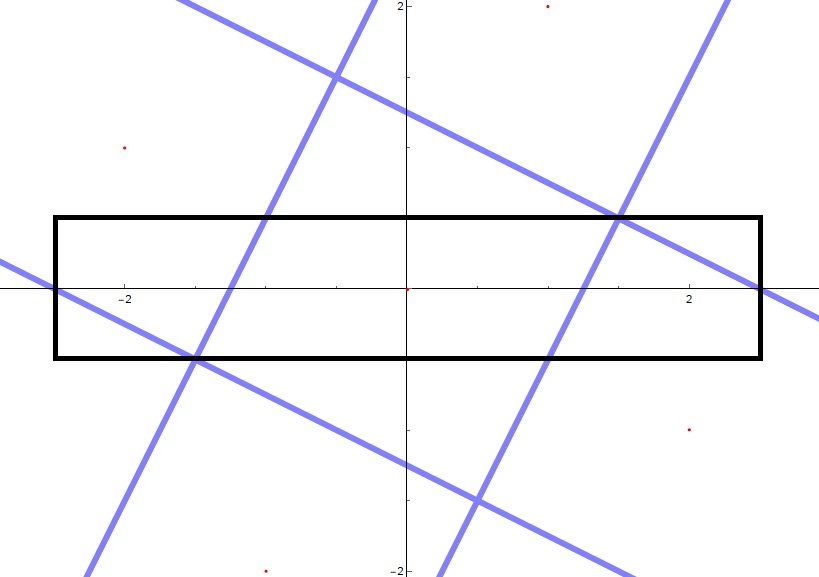}  
\caption{{Voronoi region and Babai partition of the triangular basis $\{(5,0),(3,1)\}$ }}
\label{triangbasis}
\end{center}
\end{figure}
\end{example}
Example \ref{ex1} illustrates the importance of  working with a good basis and partially explains our choice to work with a Minkowski-reduced basis. 
As mentioned above, additional motivation come from the observation that for a Minkowski-reduced basis in two dimensions, the relevant vectors are known.	

%


%

To see this, we first note that an equivalent condition for  a basis $\{v_{1}, v_{2}\}$ to be Minkowski-reduced in dimension
two is $||v_{1}||\leq||v_{2}||\leq||v_{1}\pm v_{2}||.$ Thus, we can state the following result, which was derived from the two dimensional analysis proposed in \cite{ConwaySloane:1992}.

\begin{lemma} \label{lemmathird} If a Minkowski-reduced basis is given by $\{(1,0),(a,b)\}$ then, besides the basis vectors, a third relevant vector is
\begin{equation}
\begin{cases}
(-1+a,b), & \text{if } \frac{\pi}{3} \leq \theta \leq \frac{\pi}{2} \\
(1+a,b), & \text{if } \frac{\pi}{2} < \theta \leq \frac{2\pi}{3},
\end{cases}
\end{equation}
where $\theta$ is the angle between $(1,0)$ and $(a,b).$
\end{lemma}

Note that, if $\{{v\ensuremath{_{1},v_{2}}}\}$ is a Minkowski basis
then so is $\{{-v\ensuremath{_{1},v_{2}}}\}$ and hence any lattice
has a Minkowski basis with $\frac{\pi}{2}\leq\theta\leq\frac{2\pi}{3}$.
So, if we consider the Minkowski-reduced basis as $\{(1,0),(a,b)\},$ with $a^{2}+b^{2} \geq 1$ and $-\frac{1}{2} \leq a \leq 0,$ it is possible to use Lemma \ref{lemmathird} to describe the Voronoi region of $\Lambda$ and determine its intersection with the associated Babai partition. Observe that the area of both regions must be the same and in this specific case, equal to $b.$ This means that the vertices that define the Babai rectangular partition are $\left(\pm \frac{1}{2}, \pm \frac{b}{2}\right).$ 

	In addition $\{(-1-a,-b),(1,0),(a,b)\}$ is an obtuse superbase for $\Lambda,$ so the relevant vectors that defines the Voronoi region are $\pm (1,0), \pm (a,b)$ and $\pm (-1-a,-b).$ We will choose for the analysis proposed in Theorem \ref{thmmain} only the relevant vectors in the first quadrant, i.e., $(1,0),(1+a,b), (a,b),$ due to the symmetry that a Voronoi cell has. Therefore, we can state the following result

\begin{theorem} \label{thmmain} Consider a lattice $\Lambda \subset\mathbb{R}^{2}$ with a triangular Minkowski-reduced basis $\beta=\{v_1,v_2\}=\{(1,0),(a,b)\}$   such that the angle $\theta$ between $v_{1}$ and $v_{2}$ satisfies \textup{$\frac{\pi}{2}\leq\theta\leq \frac{2\pi}{3}$.} The error probability $P_e$ for the Babai partition is given by 
\begin{equation}
P_e=F(a, b)=\frac{-a-a^{2}}{4b^{2}}=\frac{1-(1+2a)^{2}}{16b^{2}}.
\end{equation} 
\end{theorem}

\begin{proof} To calculate $P_{e}$ for the lattice $\Lambda$, we first obtain the vertices of the Voronoi region. This is done by calculating the points of intersection of the perpendicular bisectors of the  three relevant vectors $(1,0), (a,b)$ and $(1+a,b)$ (according to Lemma \ref{lemmathird}, Fig. \ref{rvectors}). Thus the vertices of the Voronoi region are given by  $\pm(\frac{1}{2},\frac{a^{2}+b^{2}+a}{2b})$,
$\pm(-\frac{1}{2},\frac{a^{2}+b^{2}+a}{2b})$ and $\pm(\frac{2a+1}{2},\frac{-a^{2}+b^{2}-a}{2})$.

\begin{figure}[h!]
\begin{center}
		\includegraphics[height=5cm]{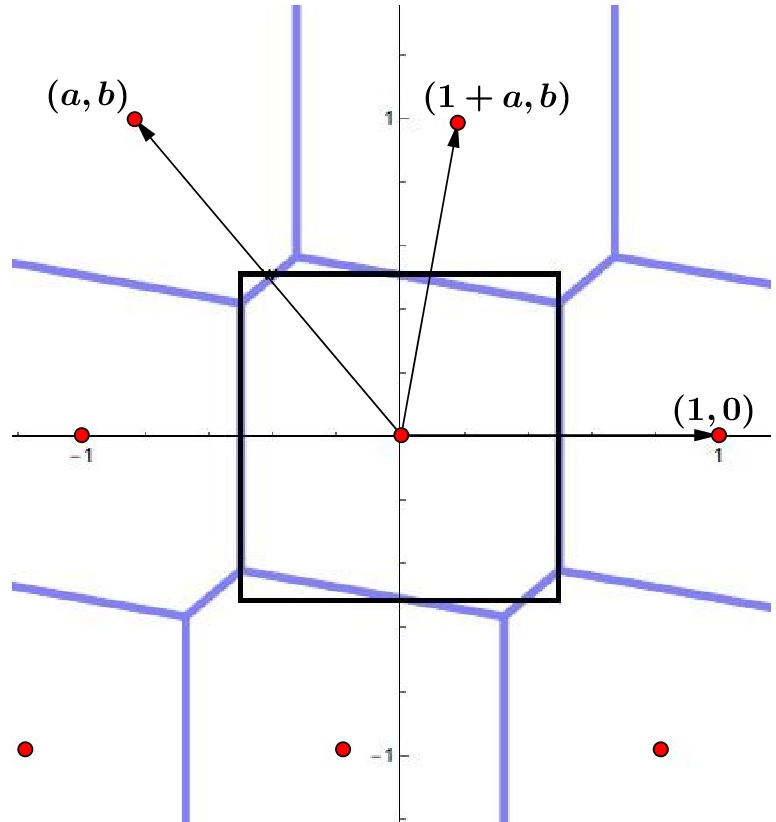}  
\caption{{Voronoi region, Babai partition and three relevant vectors}}
\label{rvectors}
\end{center}
\end{figure}

$P_{e}$ is then computed as the ratio between the area of the Babai
region which is not overlapped by the Voronoi region $\mathcal{V}(0)$ and the area $|b|$ of the Babai
region. From Fig. \ref{rvectors}, we get the error as the sum the areas of four triangles, where two of them are defined respectively by the points $\left(\frac{1}{2},\frac{b}{2}\right), \left(\frac{1}{2},\frac{a^{2}+a+b^{2}}{2b}\right), \left(\frac{a+1}{2}, \frac{b}{2}\right)$ and $\left(-\frac{1}{2},\frac{b}{2}\right), \left(-\frac{1}{2},\frac{a^{2}+a+b^{2}}{2b}\right), \left(\frac{a}{2}, \frac{b}{2}\right).$ The remaining two triangles are symmetric to these two. Therefore, the error probability is the sum of the four areas, normalized by the area of the Voronoi region $|\det(V)|=|b|$. The explicit formula for it is given by $F(a, b) =  \dfrac{1}{4} \dfrac{-a-a^{2}}{b^2}.$
\end{proof}

	We obtain the following Corollary, illustrated in Fig. \ref{contour}, from  the error probability $P_{e}=F(a,b)=\frac{1}{4} \frac{a}{b^2}(1-a)=\frac{1-(1+2a)^{2}}{16b^{2}}$ obtained in Theorem \ref{thmmain} with  $b \geq \frac{\sqrt{3}}{2}$ and $-\frac{1}{2} \leq a \leq 0.$
		
\begin{corollary} For any two dimensional lattice and a Babai partition constructed from the QR decomposition associated with a Minkowski-reduced basis where $\frac{\pi}{2} \leq \theta \leq \frac{2\pi}{3},$ we have 
\begin{equation}
0 \leq P_{e} \leq \frac{1}{12},
\end{equation}
and
\begin{itemize}
\item[a)] $P_{e}=0 \Longleftrightarrow a=0,$ i.e., the lattice is orthogonal.
\item[b)] $P_{e}=\frac{1}{12} \Longleftrightarrow (a,b)=\left(-\frac{1}{2}, \frac{\sqrt{3}}{2}\right),$ i.e., the lattice is equivalent to hexagonal lattice.
\item[c)] the level curves of $P_{e}$ are described as ellipsoidal arcs  (Figure \ref{contour}) in the region $a^{2}+b^{2} \geq 1$ and $-\frac{1}{2} \leq a \leq 0.$
\end{itemize}

\end{corollary}

\begin{figure}[h!]
\begin{center}
		\includegraphics[height=7cm]{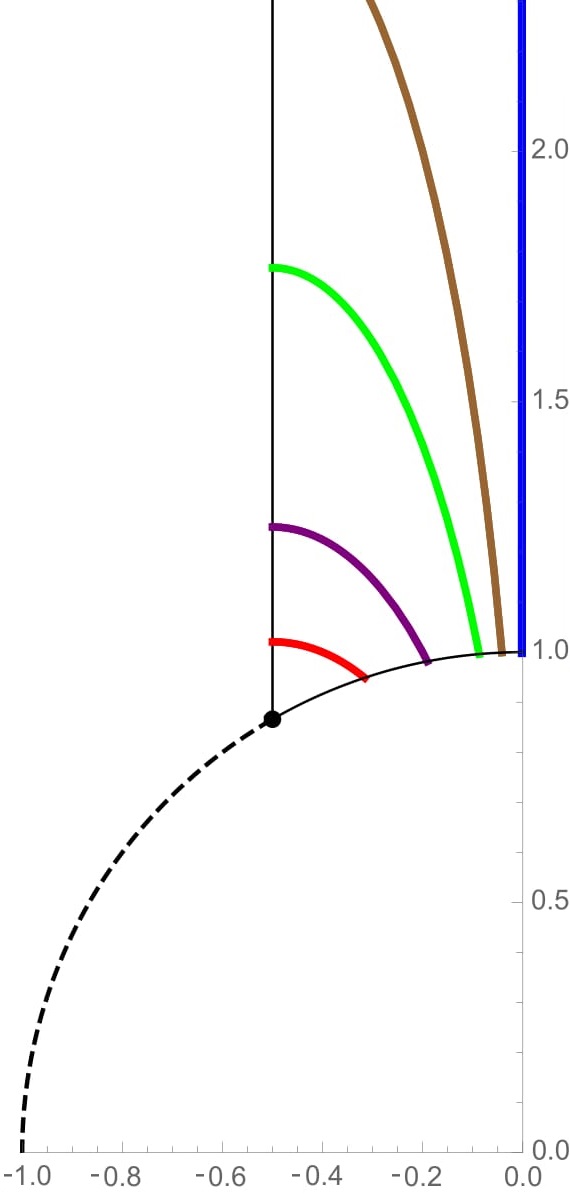}  
\caption{{Level curves of $P_{e}=k,$ in right-left ordering, for $k=0, k=0.01, k=0.02, k=0.04, k=0.06$ and $k=1/12 \approx 0.0833.$ Notice that $a$ is represented in the horizontal axis and $b$ in vertical axis. }}
\label{contour}
\end{center}
\end{figure}

\subsection{Variations with Packing Density and Angle}

	The packing density for a lattice with basis $\{(1,0),(a,b)\}$ in Minkowski-reduced form is given by $\Delta_2(a,b)=\pi/4b$ and $F(a,\Delta_2)=\frac{{\Delta^{2}_2}[1-(1+2a)^{2}]}{\pi^{2}},$ following the notation from Theorem \ref{thmmain}. For a fixed density $\Delta_2$ (fixed $b$) the error probability is decreasing with $a$ and for fixed $a,$ it is increasing with $\Delta_2$ (decreasing with $b$).
		

	So, if we consider the error probability for a given density $\Delta_2$, we have that $F(a,\Delta_2)$ is minimized by $a=a^*$, where
\begin{eqnarray}
a^*=\left\{ \begin{array}{cc} 0, &  \Delta_2 \leq \frac{\pi}{4} \ \ \ \  (b^2 \geq 1) \nonumber \\
                       -\sqrt{1-\left( \frac{\pi}{4\Delta_{2}}\right)^2},  & \frac{\pi}{4} \leq \Delta_2 \leq \frac{\pi}{2\sqrt{3}} \  \ \ \ (3/4 \leq b^2 < 1). \end{array} \right.
\end{eqnarray}
and maximized by $a= -\frac{1}{2},$ for any $\Delta_2.$ 

	Figure \ref{fig:packden2} represents the minimum error probability function $F(a, \Delta_2)$ for $\frac{\pi}{4} \leq \Delta_2 \leq \frac{\pi}{2\sqrt{3}}.$

\begin{figure}[H] 
   \centering
   \includegraphics[width=3in]{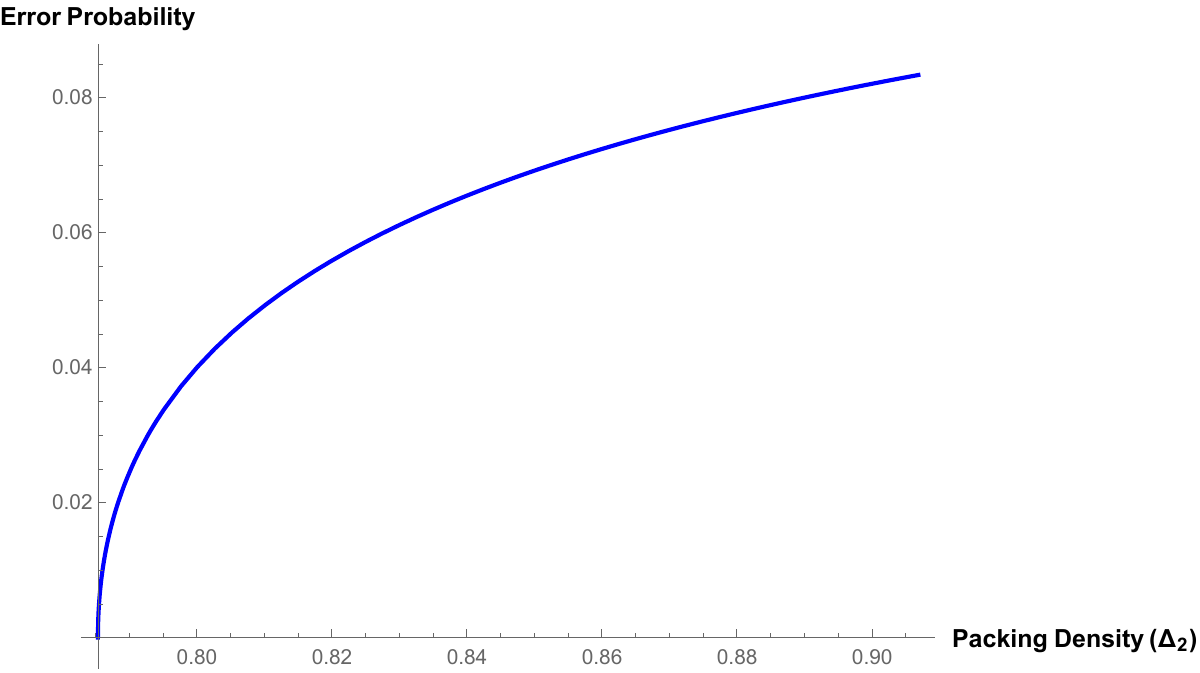} 
   \caption{Minimum error probability for given packing density assuming $\frac{\pi}{4} \leq \Delta_2 \leq \frac{\pi}{2\sqrt{3}}$}
   \label{fig:packden2}
\end{figure}

	 Note also that 
if $\rho:={\left\Vert v_{2}\right\Vert }$
and  $\theta$ is defined to be the angle between the basis vectors, then the result of Theorem \ref{thmmain} can be rewritten as 
\begin{equation}
P_{e}=H(\theta,\rho)= \frac{1-(1+2\rho \cos^{2} \theta)^{2}}{16\rho^{2} \sin^{2} \theta}.
\end{equation}
In this case, we can see that for a fixed $\rho,$ the error probability increases with $\theta,$ achieving its minimum in $\theta=0$ and maximum in $\theta= \frac{\pi}{2} + \arccos \frac{1}{2\rho}.$

%

\section{Error Probability Analysis for Arbitrary Three Dimensional Lattices} \label{sec3d}

	To analyse the error in the three dimensional case, we developed and implemented an algorithm in the software \textit{Mathematica}~\cite{Wolfram} which calculates the error probability of any three dimensional lattice, given an obtuse superbase. We assume, as we did in the two dimensional analysis, an initial upper triangular lattice basis given by $\{(1,0,0),(a,b,0),(c,d,e)\},$ where $a,b,c,d,e \in \mathbb{R}.$ It can be accomplished by performing a QR decomposition and a multiplication by a scalar factor in the original basis.
	
	It is important to remark that  the error probability is, in the general case, dependent on the basis ordering. Our algorithm searches over all orderings and determines the best one. As an example, the performance of the BCC lattice is invariant over basis ordering, due to its symmetries. On the other hand, for the FCC lattice, depending on how the basis is ordered, we can find two different error probabilities, $0.1505$ and $0.1667,$ but only $0.1505$ is tabulated. 


A detailed description of the algorithm is presented below. 

\vspace{0.3cm}

\textbf{\underline{Algorithm 1:} Error probability of the closest lattice point problem in a distributed system (three dimensional case)}

\vspace{0.3cm}

\begin{itemize}


\item[]\textbf{Voronoi region:} Provided an obtuse superbase, the vertices and faces that define the Voronoi region of $\Lambda$ are determined by Equations \eqref{eqr3} and \eqref{eqcoor}, following the method proposed by Conway and Sloane \cite{ConwaySloane:1992}. In this stage, we determine, generate and classify the correspondent Voronoi region of $\Lambda$ into one of five possibilities described in Figure \ref{5types}. 

\vspace{0.2cm}

\item[] \textbf{Babai partition:} Determine the vertices of the Babai cell. Since we have assumed a generator matrix in  upper triangular form, $\{(1,0,0),(a,b,0),(c,d,e)\}$, the vertices are:

\begin{eqnarray}
\left(\frac{1}{2}, -\frac{b}{2}, \frac{e}{2} \right), \left(\frac{1}{2}, \frac{b}{2}, \frac{e}{2} \right), \left(-\frac{1}{2}, -\frac{b}{2}, \frac{e}{2} \right), \left(-\frac{1}{2}, \frac{b}{2}, \frac{e}{2} \right), \nonumber \\
\left(\frac{1}{2}, -\frac{b}{2}, -\frac{e}{2} \right), \left(\frac{1}{2}, \frac{b}{2}, -\frac{e}{2} \right), \left(-\frac{1}{2}, -\frac{b}{2}, -\frac{e}{2} \right), \left(-\frac{1}{2}, \frac{b}{2}, -\frac{e}{2} \right).
\end{eqnarray}

\vspace{0.2cm}

\item[] \textbf{Intersection:} In this stage, using a function in Mathematica~\cite{Wolfram}, we calculate the intersection between the  Voronoi and Babai regions obtained previously. This function runs through all points that define both solids and select the coincident ones, providing in the end of the process, the vertices that determine the intersection region. We calculate then the volume of the intersection normalized by the volume of the lattice $\Lambda.$ The algorithm determines first the format of each type of Voronoi cell (Figure \ref{5types}) to simplify the calculations of the error probability. To be more specific, if all conorms $p_{ij}, (0 \leq i < j \leq 3)$ are nonzero (truncated octahedron) or if only one conorm is zero (hexa-rhombic dodecahedron) or if two collinear conorms are zero (rhombic dodecahedron), we implement the general intersection algorithm, defined as: let $v_{1}, e_{1}, f_{1}$ be, respectively, the set of vertices, edges and faces that define the Babai region of $\Lambda$ and $v_{2},e_{2}, f_{2},$ be, respectively, the set vertices, edges and faces that define the Voronoi region of $\Lambda.$ Thus, we solve:

\begin{center}
{\small \texttt{Solve \{Or $\{x,y,z\} \in e_{1}$ and $\{x,y,z\} \in f_{2}$ Or $\{x,y,z\} \in e_{2}$ and $\{x,y,z\} \in f_{1}$\}.}}
\end{center}

	The union of points $(x,y,z)$ resulting from the previous system will define the intersection of Voronoi and Babai regions of $\Lambda.$ For the two remaining cases, i.e., when we have two non-collinear zeros (hexagonal prism) we only calculate the intersection between the hexagonal basis and the rectangular basis of both prisms and when we have four zeros,  the error probability is zero.
	
\vspace{0.2cm}

\item[] \textbf{Packing density}: Finally, 
we calculate the packing density $\Delta_3$. 
\end{itemize}

\subsection{Calculations for Known Lattices}
We present results obtained by applying Algorithm 1 to some known lattices in this section.

In Fig.~\ref{fig-plot1} we have 
\begin{itemize}
\item in \textbf{\textcolor{red}{red}}, the cubic lattice $\mathbb{Z}^{3}$ with basis $\{(1,0,0),(0,1,0),(0,0,1)\};$
\item in \textbf{\textcolor{green}{green}}, the lattice with basis $\{(1,0,0),$ $(-\frac{1}{2}, -\frac{\sqrt{3}}{2}, 0),(0, 0, 1)\},$ Voronoi region:  hexagonal prism;
\item in \textbf{\textcolor{blue}{blue}}, the body-centered cubic (BCC) lattice, with basis $\{(1,0,0),(-\frac{1}{3},\frac{2 \sqrt{2}}{3},0), (-\frac{1}{3},-\frac{\sqrt{2}}{3},\sqrt{\frac{2}{3}})\}$, Voronoi region: truncated octahedron;
\item in \textbf{black}, the face-centered cubic (FCC) lattice, with basis $\{(1,0,0),(0,1,0),(-\frac{1}{2},-\frac{1}{2},\frac{1}{\sqrt{2}})\}$; Voronoi region: rhombic dodecahedron;
\item in \textbf{\textcolor{purple}{purple}}, lattice  with basis $\{(1,0,0),$ $(-\frac{1}{2}, -\frac{\sqrt{5}}{2}, 0),(0, \frac{1}{\sqrt{5}}, \frac{2}{\sqrt{5}})\},$  Voronoi region: hexa-rhombic dodecahedron.
\end{itemize}

\begin{figure}[h]
\centering
\begin{minipage}{.55\linewidth}
  \includegraphics[width=\linewidth]{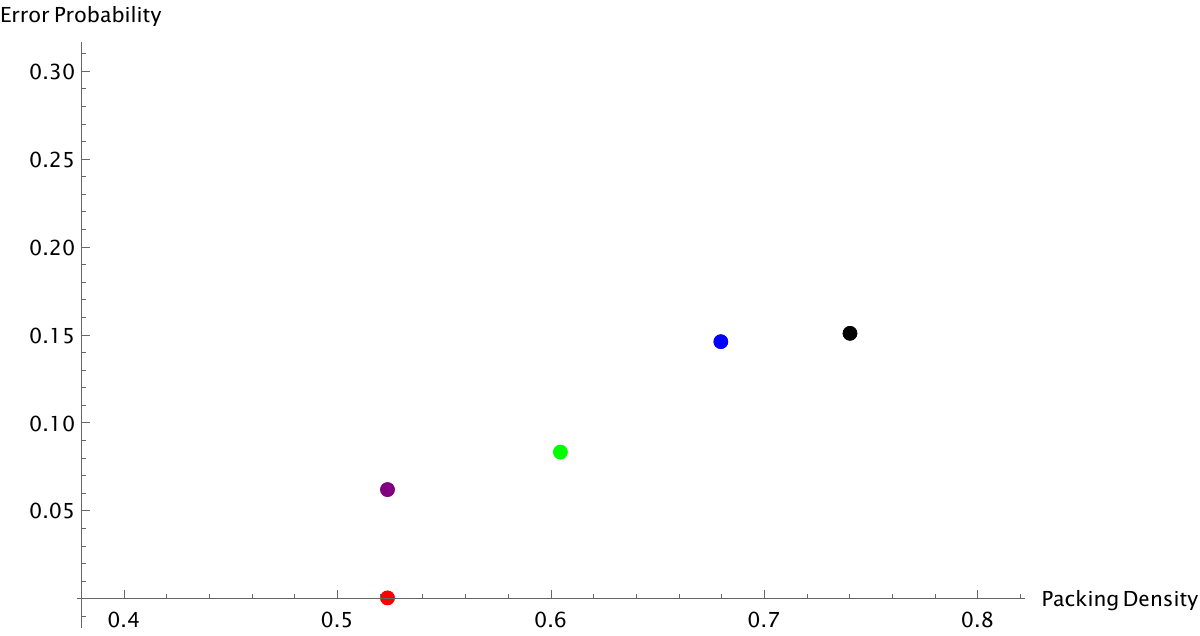}
  \caption{Performance of known lattices}
  \label{fig-plot1}
\end{minipage}
\end{figure}

Table \ref{lp} below presents some lattice performances when we run Algorithm $1.$

\begin{table}[H] 
\caption{Performance in Algorithm $1$ for known lattices} \label{lp}
\centering
 \begin{tabular}{ | c | c | c | c | c | }  
  \hline\noalign{\smallskip} 
  Lattice/Voronoi cell & Notation Table $15.6$,\cite{conwaysloane} & Conorms $(\alpha, \beta, \gamma,a,b,c)$ & $\Delta_{3}$ & $P_{e}$  \\
  \noalign{\smallskip}\hline\noalign{\smallskip}
  Cubic/ Cuboid & $1_{}1_{}1$& $(-1,-1,-1,0,0,0)$ & $0.5235$ & $0$ \\
  Hexa-rhombic dodecahedron & $2_{1}3_{1}2$ & $(-\frac{1}{2},-\frac{1}{2},-\frac{1}{2},0,-\frac{1}{2},-\frac{1}{2})$  & $0.5235$ & $0.0617$ \\
   Hexagonal prism (corresp. $A_{2}$ lattice) & $2_{-1}2_{}2$ & $(-\frac{1}{2},-\frac{1}{2},-1, 0,0,-\frac{1}{2})$  & $0.6046$ & $0.0833$ \\
   BCC/ Truncated octahedron & $3_{1}3_{1}3_{-1}$ & $(-1,-,1-,1,-1,-1,-1)$ & $0.6801$ & $0.1459$  \\
   FCC/ Rhombic dodecahedron & $2_{1}2_{1}2$ & $(-\frac{1}{2},-\frac{1}{2},0,-\frac{1}{2},-\frac{1}{2},0)$ & $0.7404$ &  $0.1505$   \\
   \noalign{\smallskip}\hline
\end{tabular}
\end{table}	 

	We remark that the error probability for the hexagonal prism is identical to the two dimensional case (see Theorem \ref{thmmain}) and cuboids  have a null error probability (when aligned to the coordinate axes). We also see that the face-centered cubic lattice, which results in the  best packing density for lattices in three dimensions, is the worst case when one considers its error probability.

\subsection{Random Lattice Selection}




	In this section we applied Algorithm 1 to lattices whose basis was chosen randomly. Specifically, we start by considering a basis at random, with the format $\{(1,0,0),$ $(a,b,0),(c,d,e)\},$ where $a,b,c,d,e$ are real numbers in the range $[-4,4].$ Then, the program tests if this basis is both an obtuse superbase and Minkowski-reduced according to Theorem \ref{minkos}. If this condition is false, another random basis is selected, until a suitable one is found. At the end of this stage, we will have a randomly chosen obtuse, Minkowski-reduced superbase for the lattice $\Lambda.$
	
	In Figure \ref{fig-plot3}, we have plotted the known points already seen in Figure \ref{fig-plot1}, together with \textbf{\textcolor{orange}{orange}} points that are associated with lattices having a packing density greater than $0.4$ randomly chosen as above.  Note that with overwhelming probability, a randomly chosen basis will have a truncated octahedron as a Voronoi region (the most general Voronoi region in three dimensions). 

\begin{figure}[h]
\centering
\begin{minipage}{.55\linewidth}
  \includegraphics[width=\linewidth]{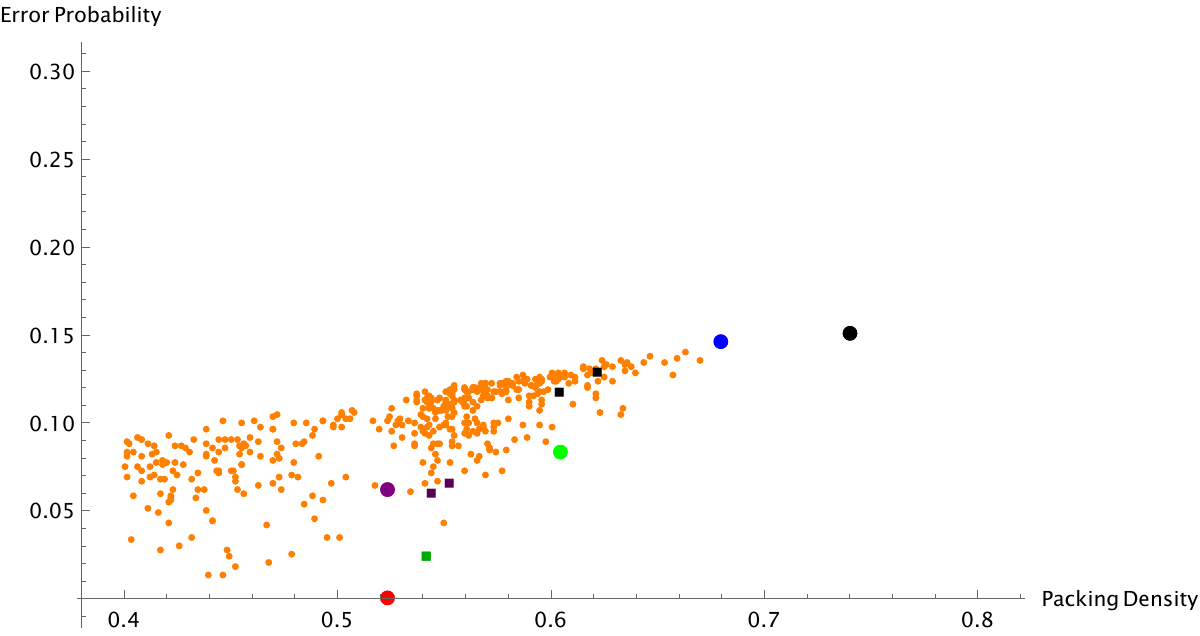}
  \caption{Comparison between random and known performances}
  \label{fig-plot3}
\end{minipage}
\end{figure}

However, by considering  conorms that are approximately zero, we can identify cases that are `almost' like one of the degenerate polyhedra. These cases are presented in Table~\ref{algorandom}, illustrated as square points in Figure \ref{fig-plot3}, where the color characterizes the cell type, following the notation of Figure \ref{fig-plot1}.

\begin{table}[H] 
\caption{Performance in Algorithm $1$ for random lattices} \label{algorandom}
\centering
 \begin{tabular}{ | c | c | c | c | }  
  \hline\noalign{\smallskip} 
  
  (Aproximate) Voronoi cell & Conorms $(\alpha, \beta, \gamma,a,b,c)$ & $\Delta_{3}$ & $P_{e}$  \\
  \noalign{\smallskip}\hline\noalign{\smallskip}
   Hexa-rhombic dodecahedron & $(-0.4447,-0.7089,-0.7596,-0.0007,-0.3055,-0.2903)$ & $0.5441$ & $0.0592$ \\
   Hexa-rhombic dodecahedron & $(-0.3128,-0.7110,-0.6812,-0.0005,-0.4535,-0.2884)$  & $0.5527$ & $0.0652$ \\
   Rhombic-dodecahedron & $(-0.0574,-0.5159,-0.7771,-0.0041,-0.4708,-0.4798) $& $0.6044$ & $0.1169$ \\
   Rhombic-dedecahedron & $(-0.5280,-0.05218,-0.6273,-0.4968,-0.0650,-0.4509) $& $0.6220$ & $0.1280$ \\
   Hexagonal prism & $(-0.5246,-0.9048,-0.6788,-0.0201,-0.4024,-0.0750) $& $0.5417$ & $0.0237$ \\
      \noalign{\smallskip}\hline
\end{tabular}
\end{table}

\subsection{Remarks}
We conjecture, after $5000$ trials that:
\begin{conjecture} For any three dimensional lattice and a Babai partition constructed from the QR decomposition associated with an obtuse superbase which is also Minkowski-reduced,
\begin{equation}
0 \leq P_{e} \leq 0.1505.
\end{equation}
\end{conjecture}


	Compared with the two dimensional case, we have an increase of $7.32\%$ in the conjectured bound for the error probability and we expect this number to grow more as the dimension increases. 
	
	We also conjecture, assuming a more "spherical shape" for Voronoi regions of densest lattices the following
	
\begin{conjecture} The worst error probability for a lattice in dimension $n$ is achieved by the densest lattice and it tends to one when $n$ goes to infinity.
\end{conjecture}


\section{Rate Computation for Constructing a Babai Partition for arbitrary $n>1$}
\label{sec:babairate}
Communication protocols are presented for the centralized and interactive model along with associated rate calculations in the limit as $\alpha \rightarrow 0$.   

\subsection{Centralized Model}
We now describe the transmission protocol  $\Pi_{c}$ by which the nearest plane lattice point can be determined at the fusion center $F$.
Let $v_{ml}/v_{mm}=p_{ml}/q_{ml}$ where $p_{ml}$ and $q_{ml}>0$ are relatively prime. Note that we are assuming the generator matrix is such that the aforementioned ratios are rational, for $l>m.$ Let $q_m=l.c.m \ \{q_{ml}, l>m\}$, where $l.c.m$ denotes the least common multiple of its argument. By definition $q_n=1$.
\begin{protocol} (Transmission, $\Pi_{c}$). Let $s(m)\in\{0,1,\ldots,q_m-1\}$ be the largest $s$ for which $[x_m/v_{mm}-s/q_m]=[x_m/v_{mm}]$. Then node $m$ sends 
$\tilde{b}_m=[x_m/v_{mm}]$ and $s(m)$ to $F$, $m=1,2,\ldots,n$ (by definition $s(n)=0$).
\end{protocol}

Let $\overline{b}=(b_1,b_2,\ldots,b_n)$ be the coefficients of $\lambda_{np}$, the Babai point.
\begin{theorem}\label{thmcost}
The coefficients of the Babai point $\overline{b}$ can be determined at the fusion center $F$ after running transmission protocol $\Pi_c$.
\end{theorem}
\begin{proof}
Observe that each coefficient of $\overline{b}$ is given by
\begin{eqnarray}
\lefteqn{b_m=} \ \ \ \ \ \ \ \ \left[ \frac{x_m-\sum_{l=m+1}^{n} b_{l}v_{m,l}}{v_{mm}}\right],
m=1,2,\ldots,n,
\end{eqnarray}
which is written in terms of  $\{z \}$ and $\lfloor z \rfloor$, the fractional and integer parts of real number $z$, resp.,  ($z=\lfloor z\rfloor +\{z\}$, $0 \leq \{z\} < 1$) by
\begin{eqnarray}
\lefteqn{b_m=} \ \ \ \ \ \ \ \ \ \left[ \frac{x_m}{v_{m}}-\left\{\frac{\sum_{l=m+1}^{n} b_{l}v_{ml}}{v_{mm}}\right\} \right]-  \left\lfloor \frac{\sum_{l=m+1}^{n} b_{l}v_{ml}}{v_{mm}}\right\rfloor, \ \ \ \ m=1, 2, \ldots,n.
\end{eqnarray}
Since the fractional part in the above equation is of the form $s/q_m$, $s \in \{0,1,\ldots,q_m-1\}$, where $q_m$ is defined above, it follows that  $0 \leq s/q_m < 1$. Thus
\begin{eqnarray} \label{decision}
\lefteqn{{b}_m=} \ \ \ \ \ \ \ \ \ \left\{ \begin{array}{cc}
\tilde{b}_m-\left\lfloor \frac{\sum_{l=m+1}^{n}{b}_{l}v_{ml}}{v_{mm}}\right \rfloor, & s\leq s(m),  \\
\tilde{b}_m-\left\lfloor \frac{\sum_{l=m+1}^{n}b_{l}v_{ml}}{v_{mm}}\right \rfloor-1,  & s >  s(m).
\end{array}
\right.
\end{eqnarray}
can be computed in the fusion center $F$ in the order $m=n,n-1,\ldots,1$.
\end{proof}
\begin{corollary}\label{corocost}
The rate required to transmit $s(m)$, $m=1,2,\ldots,n-1$ is no larger than $\sum_{i=1}^{n-1}\log_2(q_i)$ bits.
\end{corollary}
Thus the total rate for computing the Babai point at the fusion center $F$ under the centralized model is no larger than $\sum_{i=1}^nh(p_i)-\log_2|\det V| -n \log_2(\alpha)+\sum_{i=1}^{n-1}\log_2(q_i)$ bits, where $h(p_i)$ is the differential entropy of random variable $X_i$, and scale factor $\alpha$ is small. Thus the incremental cost due to the $s(m)$'s does not scale with $\alpha$. However when $\alpha$ is small, this incremental cost can be considerable, if the lattice basis is not properly chosen as we will see in further examples.

	This rate computation can be visualized geometrically and under the light of the decoding in orthogonal lattices. Consider a lattice $\Lambda \subset \mathbb{R}^{n}$ generated by $\{v_1, v_{2}, \dots, v_{n}\},$ where we want to decode under the constraints proposed by the centralized model, a real vector $x=(x_1, x_2, \dots, x_{n}).$ We construct an associated orthogonal lattice $\Lambda' \subset \mathbb{R}^{n}$ whose basis vectors are $\{\underbrace{(v_{11}, 0, \dots, 0)}_{v_{1}'}, \dots, \underbrace{(0, 0, \dots, v_{nn})}_{v_{n}'}\},$ where $v_{ii}, 1 \leq i \leq n$ are the diagonal elements from the original generator matrix of $\Lambda.$ Observe that the Voronoi region of $\Lambda'$ corresponds to the Babai partition not aligned achieved without sending any extra bit in this model.
	
	The idea is to decode in the orthogonal associated lattice $\Lambda',$ which is a simple process and after that, recover the original approximate closest lattice point in $\Lambda.$ At the end, we aim to prove that this process is equivalent to sending the extra bits and with this information, decide between the cases described in Equation (\ref{decision}).
	
	Initially, we can notice that these Babai partitions in the space follow a cyclic behavior, i.e., after exactly $\prod_{m=1}^{n-1} q_{m},$ where $q_{m}= l.c.m \ \{q_{ml}, l>m\}$ shifts, it comes back to the original setting. This number, when calculated as a rate, corresponds precisely to the upper bound we have for the extra bits, introduced in Corollary \ref{corocost}. 
	
\begin{example} Consider a lattice $\Lambda$ generated by $\{(1,0),(2/5,2)\}$ and $\Lambda'$ generated by $\{(1,0),(0,2).\}$ In this case, $q_{1}=5$ and there are $q_{1}$ distinct settings in the plane one need to analyze.  After $q_{1}$ shifts, the Voronoi aligned partition around lattice points in the form $(0,\kappa),$ $\kappa \in \mathbb{Z},$ starts to be repeated, as illustrated in Figure \ref{vv}.

\begin{figure}[H]
\begin{center}
		\includegraphics[height=5.5cm]{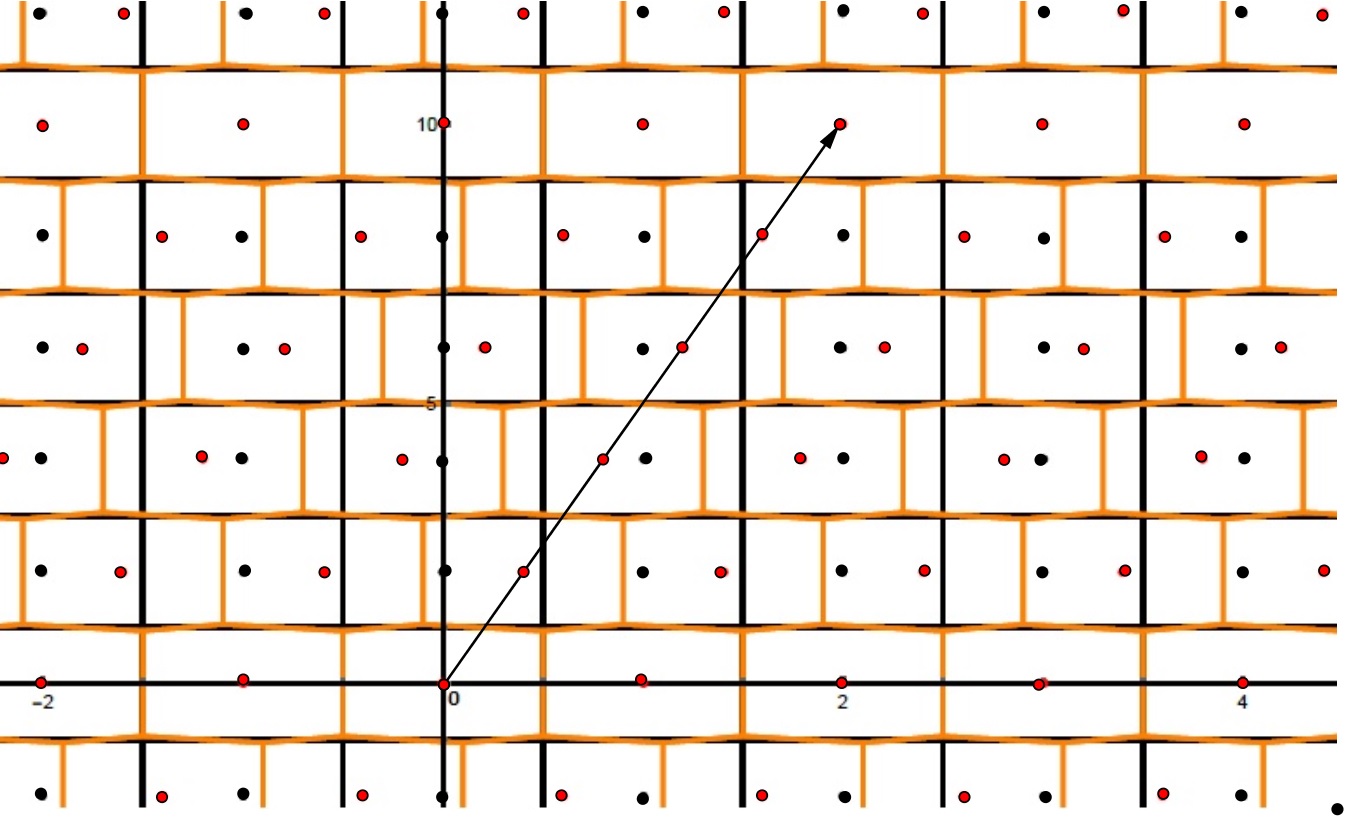}  
\caption{{Voronoi partition of $\Lambda$ in orange and Voronoi partition of $\Lambda'$ (Babai) in black}}
 \label{vv}
\end{center}
\end{figure}

	This situation can be seen as a "modulo $q_{m}$" operation, where each class is represented uniquely in the space.
\end{example}

	The vector $\tilde{b}=(\tilde{b}_{1}, \dots, \tilde{b}_{n}),$ with $\tilde{b}_{i}=\left [ x_{i}/ v_{ii} \right],$ is such that $||x-V'\tilde{b}||$ is minimum, where $V'$ has the vectors $v_{1}', \dots, v_{2}'$ on its columns. It means that $\tilde{b}$ decodes $x \in \mathbb{R}^{n}$ in the associated lattice $\Lambda'$ and we want to use this information to decode approximately in $\Lambda.$ Clearly, $\tilde{b}_{n}=b_{n}$ always.

	In a general two dimensional case, for a matrix in the form $V=\begin{pmatrix}
	v_{11} & v_{12} \\
	0 & v_{22}
	\end{pmatrix},$ we consider $\tilde{b}_{2}=b_{2}.$ Indeed, this fact is always true because essentially, we want to write the vector $(x_{1},x_{2})$ in terms of the basis $\{v_{1},v_{2}\}.$ So, we have that
\begin{eqnarray}
\begin{pmatrix}
	v_{11} & v_{12} \\
	0 & v_{22}
	\end{pmatrix} \begin{pmatrix}
	{b}_{1} \\
	{b}_{2} 
	\end{pmatrix} = \begin{pmatrix}
	x_{1} \\
	x_{2} 
	\end{pmatrix}.
\end{eqnarray}

	To recover the aligned Babai partition, one aims to find:
%
%

{\small \begin{eqnarray} \label{eq2d}
b_{2} &=& \left[ \frac{x_{2}}{v_{22}} \right],\nonumber \\
v_{11}b_{1}+v_{12}b_{2}= x_{1} \Rightarrow b_{1} &=& \left[ \frac{x_{1}}{v_{11}} -\frac{v_{12}}{v_{11}}b_{2} \right] = \left[ \frac{x_{1}}{v_{11}} - \frac{p_{12}}{q_{12}} b_{2} \right] \nonumber \\
&=& \left[ \underbrace{\left[ \frac{x_1}{v_{11}} \right]}_{\tilde{b}_{1}} + \overline{x}_{1} - \left\{ \frac{p_{12}}{q_{12}}b_{2} \right\} \right]  - \underbrace{\left\lfloor \frac{p_{12}}{q_{12}}b_{2} \right\rfloor}_{(b_{2}p_{12} \mod q_{1})},
\end{eqnarray}}
where $-\frac{1}{2} < \overline{x}_{1} < \frac{1}{2}.$ Geometrically, this operation means that we are bringing the analysis in each case to one of the $\{0,1, \dots, q_{1}-1\}$ classes and correcting it by a factor of $(b_{2}p_{12} \mod q_{1}),$ which represents the translation occurred to the lattice point.


\begin{example} Figure \ref{vv} has represented in black the lattice points of $\Lambda'$ and in red the lattice points of $\Lambda,$ which are the ones we want to recover at the end of the process. We can immediately notice that the correction we need to take in account depends on where $x_{1}$ is located in the plane. For example, if $\tilde{b}_{2}=3=b_{2}$ and $-\frac{1}{2} < x_{1} < \frac{1}{2},$ then
\begin{equation}
b_{1}=\begin{cases}
\tilde{b}_{1}- 1, & \ \text{if} \ -\frac{1}{2}+ \frac{1}{5} = -\frac{3}{10} < x_{1}-[x_{1}] < \frac{1}{2} \\
(\tilde{b}_{1}- 1) -1, & \ \text{if} \ -\frac{1}{2} < x_{1}-[x_{1}] \leq -\frac{3}{10}.  \\
\end{cases}
\end{equation}

	In a more general setting, according to Equation (\ref{eq2d}), we have that:
\begin{equation}
b_{1}=\begin{cases}
\tilde{b}_{1}  - (b_{2}p_{12} \mod q_{1}) , & \ \text{if} \ -\frac{v_{11}}{2} + \left\{ \frac{p_{12}}{q_{12}b_{2}} \right\} <\frac{x_{m}}{v_{mm}}-  \left[ \frac{x_{m}}{v_{mm}}\right] < \frac{v_{11}}{2} \\
\tilde{b}_{1} - (b_{2}p_{12} \mod q_{1}) - 1, & \ \text{if} \ -\frac{v_{11}}{2} < \frac{x_{m}}{v_{mm}} - \left[ \frac{x_{m}}{v_{mm}} \right]  \leq  -\frac{v_{mm}}{2} + \left\{\frac{p_{12}}{q_{12}b_{2}} \right\}.  \\
\end{cases}
\end{equation}
\end{example}

	This analysis can be also described for the $n-$dimensional case, where we aim to find the Babai point $\overline{b}=(b_{1}, \dots, b_{n}),$ as
{\small \begin{equation}
b_{m}=\begin{cases}

\tilde{b}_{m}  - (\sum_{l=m+1}^{n}{b}_{l}p_{ml}\hat{q}_{ml} \mod q_{m}) , & \ \text{if} \ -\frac{v_{mm}}{2} + \left\{ \frac{\sum_{l=m+1}^{n}{b}_{l}v_{ml}}{v_{mm}} \right\} <\frac{x_{m}}{v_{mm}}-  \left[ \frac{x_{m}}{v_{mm}}\right] < \frac{v_{mm}}{2} \\

\tilde{b}_{m} - (\sum_{l=m+1}^{n}{b}_{l}p_{ml}\hat{q}_{ml} \mod q_{m}) - 1, & \ \text{if} \ -\frac{v_{mm}}{2} < \frac{x_{m}}{v_{mm}} - \left[ \frac{x_{m}}{v_{mm}} \right]  \leq  -\frac{v_{mm}}{2} + \left\{\frac{\sum_{l=m+1}^{n}{b}_{l}v_{ml}}{v_{mm}} \right\}, \\
\end{cases}
\end{equation}}
where $\hat{q}_{ml}={q_{m}}/{q_{ml}}.$ Therefore, the cost of analyzing all the classes is no larger than  $\sum_{m=1}^{n-1}\log_2(q_m),$ as stated in Corollary \ref{corocost}.

	The following example illustrated how the method proposed in Theorem \ref{thmcost} works in two and three dimensions and also explore a case where this cost could be large.

\begin{example} Consider the hexagonal $A_{2}$ lattice generated by 
$$V=\begin{pmatrix}
1 & \frac{1}{2} \\
0 & \frac{\sqrt{3}}{2}
\end{pmatrix}.$$ 
 The basis vectors are already Minkowski-reduced and applying what we described above we have that the coefficients $b_{2}$ and $b_{1}$ are given respectively by
{\small \begin{equation}
b_{2}=\left[\frac{x_{2}}{v_{22}} \right]=\left[\frac{2}{\sqrt{3}}x_{2} \right]
\end{equation} }
and
{\small \begin{eqnarray}
b_{1}&=&\left[\frac{x_{1}}{v_{11}} - \left\{\frac{b_{2}v_{21}}{v_{11}} \right\} \right]-\left\lfloor\frac{b_{2}v_{21}}{v_{11}} \right\rfloor \\
&=& \left[ x_{1} -  \left\{\left[\frac{2}{\sqrt{3}}x_{2} \right] \frac{1}{2} \right\} \right] - \left\lfloor\left[ \frac{2}{\sqrt{3}}x_{2} \right] \frac{1}{2} \right\rfloor.
\end{eqnarray} }

	Hence, for any real vector $x=(x_1, x_2)$ we have $\left\{ \left[ \frac{2}{\sqrt{3}}x_{2} \right] \frac{1}{2} \right\}= \frac{s}{q}$, with $q=2$ and $s \in \{0,1\}$.  Node one must then send the largest integer $s(1)$ in the range $\{0,1\}$ for which $\left[x_{1}-\frac{s(1)}{q_{1}}\right]=[x_{1}]$ and $s(1)=0$ or $s(1)=1$ depending on the value that $x_{1}$ assumes.
	
	The cost of this procedure, according to Corollary \ref{corocost}, is no larger than $\log_{2}q_{1}=1$ bit. Thus the cost of constructing the nearest plane partition for the  hexagonal lattice is at most one bit.

\end{example}

	Nevertheless, this rate could be potentially large as the next example illustrates.

\begin{example} Suppose a lattice generated by 
$$V=\begin{pmatrix}
1 & \frac{311}{1000} \\
0 & \frac{101}{100}
\end{pmatrix}.$$ 
One can notice that the basis vectors are already Minkowski-reduced. Using the theory developed above we have that
{\small \begin{equation}
b_{2}=\left[\frac{x_{2}}{v_{22}} \right]=\left[\frac{100}{101}x_{2} \right]
\end{equation} }
and
{\small \begin{eqnarray}
b_{1}&=&\left[\frac{x_{1}}{v_{11}} - \left\{\frac{b_{2}v_{21}}{v_{11}} \right\} \right]-\left\lfloor\frac{b_{2}v_{21}}{v_{11}} \right\rfloor \\
&=& \left[ x_{1} -  \left\{\left[\frac{100}{101}x_{2} \right] \frac{311}{1000} \right\} \right] - \left\lfloor\left[ \frac{100}{101}x_{2} \right] \frac{311}{1000} \right\rfloor.
\end{eqnarray} }

	Consider, for example, $x=(1, 1)$ then we have that $\left\{\left[\frac{100}{101}x_{2} \right] \frac{311}{1000} \right\}=\frac{311}{1000}=\frac{s}{q}.$ In this purpose, node one sends the largest integer $s(1)$ in the range $\{0,1, \dots, 999\}$ for which $\left[x_{1}-\frac{s(1)}{q_{1}}\right]=[x_{1}]$ and we get $s(1)=500.$ 
	
	This procedure will cost no larger than $\log_{2}q_{1}=\log_{2}1000 \approx 9.96$ and in the worst case, we need to send almost $10$ bits to achieve Babai partition in the centralized model. 

\end{example}

\begin{example} Consider the three dimensional body centered cubic (BCC) lattice with generator matrix given by
\begin{equation}
V=\begin{pmatrix}
1 & -\frac{1}{3} & -\frac{1}{3} \\
0 & \frac{2\sqrt{2}}{3} & -\frac{\sqrt{2}}{3} \\
0 & 0 & \sqrt{\frac{2}{3}}
\end{pmatrix}.
\end{equation}
We already checked that this basis generates an obtuse superbase, which is also Minkowski-reduced according to Theorem \ref{minkos}. Thus, in order to align this Voronoi region with Babai partition in its best way, we need to calculate the Babai point given by $(b_1, b_2, b_3)$ described below.

{\small \begin{eqnarray} \label{b3}
b_{3} &=& \left[ \sqrt{\frac{3}{2}} x_{3} \right],
\end{eqnarray}}
{\small \begin{eqnarray} \label{b2}
b_{2}&=&\left[ \frac{3}{2\sqrt{2}}x_{2} + \left\{ \frac{1}{2} b_{3} \right\} \right]+\left\lfloor\frac{1}{2} b_{3} \right\rfloor \nonumber \\
&=& \left[ \frac{3}{2\sqrt{2}}x_{2} + \left\{\frac{1}{2}\left[ \sqrt{\frac{3}{2}} x_{3} \right] \right\} \right] +\left\lfloor\frac{1}{2} \left[ \sqrt{\frac{3}{2}} x_{3} \right] \right\rfloor, 
\end{eqnarray}}
and
{\small \begin{eqnarray}
b_{1}&=&\left[ x_1 + \left\{ \frac{1}{3} b_{2} + \frac{1}{3} b_{3} \right\} \right]+\left\lfloor\frac{1}{3} b_{2} + \frac{1}{3} b_{3} \right\rfloor, 
\end{eqnarray}}
where $b_{2}$ and $b_{3}$ are integers previously defined in Equations (\ref{b3}) and (\ref{b2}), respectively. 

	Hence, for any real vector $x=(x_1, x_2,x_3)$ we have two nodes that should send extra information, nodes $2$ and $1,$ according to the following description
\begin{equation}
\begin{cases}
\text{Node 2:} \ \left\{ \frac{1}{2} b_{3} \right\} = \frac{s(2)}{q(2)}, \ q(2)=2 \ \text{then} \ s(2)=0 \ \text{or} \ 1 \\
\text{Node 1:} \ \left\{ \frac{1}{3} b_{2} + \frac{1}{3} b_{3} \right\} = \frac{s(1)}{q(1)}, \ q(1)=3 \ \text{then} \ s(1)=0,1 \ \text{or} \ 2 .\\
\end{cases}
\end{equation}
	
	Observe that the values of $s(1)$ and $s(2)$ are calculated here in a general way, however, they exact values depend on $x_{1}$ and $x_{2},$ respectively. Therefore, the total rate to send $s(1)$ and $s(2)$ to the fusion center is 
\begin{eqnarray}
\log_{2}2 + \log_{2}3 \approx 2.5859 \approx 3 \ \text{bits}.
\end{eqnarray}
\end{example}

%
%

The analysis here points to the importance of the number-theoretic structure of the generator matrix  $V$ in determining the communication requirements for computing $x_{np}$. \\

\subsection{Interactive Model}
For $i=n,n-1,\ldots,1$, node  $S_i$ sends $U_i=\left[ (X_i-\sum_{j=i+1}^n\alpha v_{ij}U_j)/\alpha v_{ii}\right]$ to all other nodes. The total number of bits communicated is given by $R=(n-1)\sum_{i=1}^n H(U_i|U_{i+1},U_{i+2},\ldots,U_n)$. For $\alpha$ suitably small, and under the assumption of independent $X_i$, this rate can be approximated by $R=(n-1)\sum_{i=1}^n h(p_i)-\log_2(\alpha v_{ii})$. Normalizing so that $V$ has unit determinant we get $R=(n-1)\sum_{i=1}^nh(p_i)-n(n-1) \log_2(\alpha)$.

\section{Conclusion and future work}	
\label{secCC}

	We have investigated the closest lattice point problem in a distributed network, under two communication models, centralized and interactive. By exploring the nearest plane (Babai) partition for a given Minkowski-reduced basis, we have determined a closed form for the error probability in two dimensions. For the three dimensional case, using an obtuse superbase, we have estimated computationally for random lattices the worst error probability. The number of bits that nodes need to send in both models (centralized and interactive) to achieve the rectangular nearest plane partition was computed.	
	
	Further problems to be investigated are regarding similar results to be derived for greater dimensions, for example, it may be possible to generalize the results presented here to families $A_{n} $ and $D_{n}$ lattices, for which reduced form bases are already available. Another direction is to analyse the rate computation for the centralized and interactive modes when the standard Viterbi algorithm based on orthogonal sublattices is considered.
	
%
%
	
\section{Acknowledgment}
CNPq (140797/2017-3, 312926/2013-8) and FAPESP (2013/25977-7) supported the work of MFB and SIRC. VV was supported by  CUNY-RF and CNPq (PVE 400441/2014-4). MFB would like to thank Nelson G. Brasil for meaninful discussions and contributions regarding to the computational implementation.

\end{document}